\documentclass[letter,11pt]{article}
\usepackage{fullpage}
\usepackage{amsmath}
\usepackage{amsfonts}
\usepackage{amsthm}
\usepackage{amssymb}
\usepackage{dsfont}
\usepackage{graphicx}
\usepackage{caption}
\usepackage{subcaption}

\newtheorem{theorem}{Theorem}
\newtheorem{lemma}[theorem]{Lemma}

\newtheorem{cor}[theorem]{Corollary}

\newcommand{\bigO}{\mathcal{O}}

\DeclareMathOperator*{\argmax}{\arg\!\max}

\DeclareMathOperator*{\E}{\mathbb{E}}
\DeclareMathOperator{\R}{\mathbb{R}}
\DeclareMathOperator*{\one}{\mathds{1}}
\DeclareMathOperator*{\vv}{\mathbf{v}}
\DeclareMathOperator*{\q}{\mathbf{q}}
\DeclareMathOperator*{\rr}{\mathbf{r}}

\DeclareMathOperator{\F}{\mathbf{F}}
\DeclareMathOperator*{\HH}{\mathbf{H}}

\DeclareMathOperator*{\ud}{\text{d}}

\DeclareMathOperator{\calV}{\mathcal{V}}
\DeclareMathOperator{\I}{\mathbf{I}}
\DeclareMathOperator{\area}{\text{area}}
\DeclareMathOperator{\supp}{\text{supp}}

\allowdisplaybreaks
\title{Bounds on the revenue gap of linear posted pricing\\for selling a divisible item}

\author{Ioannis Caragiannis\thanks{Department of Computer Science, Aarhus University, Denmark. Email: {\tt iannis@cs.au.dk}.} \and Zhile Jiang\thanks{Department of Computer Science, Aarhus University, Denmark. Email: {\tt zhile@cs.au.dk}.} \and Apostolis Kerentzis\thanks{Tactical Air Force, Hellenic Air Force, Greece. Email: {\tt apostolos.kerentzis@haf.gr}.}}
\date{}

\begin{document}
\maketitle

\begin{abstract}
Selling a perfectly divisible item to potential buyers is a fundamental task with apparent applications to pricing communication bandwidth and cloud computing services. Surprisingly, despite the rich literature on single-item auctions, revenue maximization when selling a divisible item is a much less understood objective. We introduce a Bayesian setting, in which the potential buyers have concave valuation functions (defined for each possible item fraction) that are randomly chosen according to known probability distributions. Extending the sequential posted pricing paradigm, we focus on mechanisms that use linear pricing, charging a fixed price for the whole item and proportional prices for fractions of it. Our goal is to understand the power of such mechanisms by bounding the gap between the expected revenue that can be achieved by the best among these mechanisms and the maximum expected revenue that can be achieved by any mechanism assuming mild restrictions on the behavior of the buyers. Under regularity assumptions for the probability distributions, we show that this revenue gap depends only logarithmically on a natural parameter characterizing the valuation functions and the number of agents. Our results follow by bounding the objective value of a mathematical program that maximizes the ex-ante relaxation of optimal revenue under linear pricing revenue constraints.
\end{abstract}

\section{Introduction}
Selling a single item to potential buyers is probably the most fundamental problem in microeconomics, with amazing related discoveries during the last sixty years. In the most standard model, there are several potential buyers (the {\em agents}), with {\em private valuations} for the item. The celebrated {\em Vickrey auction}~\cite{V61} allocates the item to the {\em highest bidder} and charges her the second-highest bid as price. At least in theory, it achieves two important goals. First, it motivates the agents to act {\em truthfully} and report their true valuations as bids. Secondly, it maximizes {\em social welfare}, in the sense that it allocates the item to the agent who values it the most.

In contrast, for the equally important objective of {\em revenue maximization}, no solution is possible unless some information is available for the agent valuations. In the classical {\em Bayesian setting}, which is the most appealing for studying revenue maximization, each agent is assumed to have a random private valuation drawn from a probability distribution. Information about the probability distributions is known to the seller, who can use it to adjust the auction parameters and maximize her expected revenue. Under a regularity assumption about the agents' distributions, the revenue-optimal auction allocates the item to the agent with the highest non-negative {\em virtual bid} (if any) and charges her the minimum ``critical'' bid she could still use to win the item. Among other important results, Myerson's seminal paper~\cite{M81} presents the related theory for revenue-optimal mechanism design, which can be applied to any {\em single-parameter} environment where the private valuation of each agent is just a scalar.

Inevitably, revenue-optimal auctions can be {\em complicated} and {\em counter-intuitive}. For example, as the virtual bid of an agent depends on the distribution from which she draws her valuations, they may discriminate among agents, setting different reserve prices for each of them. Also, the winning bid is not necessarily the highest one. So, the auctions that are used in practice are much simpler. A standard practice is to run a second-price auction with an {\em anonymous reserve} (i.e., common for all agents). Another even simpler class of mechanisms, known as {\em sequential posted pricing}~\cite{CHMS10}, approaches the agents in a predefined order and makes a take-it-or-leave-it offer to each of them until the item is bought. When approached, a potential buyer can either accept to buy the item at the proposed price or refuse to buy if the proposed price exceeds her valuation for the item. In general, sequential posted pricing can become notoriously complex (pricing in the airline industry is an annoying example from practice) as the best possible price for an agent can depend on statistical information about the agents' valuations and on the decisions of agents that were approached previously. 

However, sequential posted pricing is simple and well-understood when the same {\em anonymous} price (computed using statistical information about the valuations) is proposed to all agents. Constant approximations to optimal revenue (or a constant {\em revenue gap}) by sequential posted pricing with an anonymous price is possible under a regularity assumption for the valuations. Such statements (e.g., in~\cite{JLQTX19}) usually read as follows: for any set of buyers who draw their valuations from independent and regular probability distributions, there is a price depending only on these distributions that yields an expected revenue that is a constant fraction of the expected revenue returned by Myerson's optimal auction.

In this work, we deviate from the above setting and assume that the item for sale is {\em perfectly divisible}. Different agents can get fractions of the item, while some fraction of the item can stay unsold. Agent preferences and behavior are much richer now. Each agent is still interested in obtaining the whole item but gets value from fractions of it as well. In particular, each agent has a private {\em valuation function} that indicates her value for fractions of the item between $0$ and $1$. Typically, these functions are non-decreasing and concave, indicating non-increasing non-negative marginal value. Pricing of a divisible item can be very complicated in this setting. For example, refining sequential posted pricing, we can imagine a seller proposing a {\em menu of prices} to the agent, with a different price for different fractions of the item. In its extreme, the scheme could include different {\em pricing functions}, discriminating among agents. Then, taking into account her valuation function when facing a proposed pricing function, the agent can demand a desired fraction of the item that maximizes her utility.

We are interested in simple sequential posted pricing mechanisms that specifically use {\em linear pricing}. In particular, we restrict pricing to a fixed price for the whole item and proportional prices for fractions. Do such mechanisms have nearly-optimal revenue? Can we have statements like the one for anonymous pricing above? We will address these very challenging questions, making several conceptual and technical contributions.

\subsection{Our contribution and techniques} 
Our first conceptual contribution is an extension of the standard Bayesian setting. We assume that each agent has a random non-decreasing concave valuation function. An important assumption we make is that the derivative (or {\em marginal}) of the valuation function of each agent at each fraction point follows a regular probability distribution. All the information about the probability distributions of each agent is known to the seller.

As our setting is very far from the single-parameter environment of an indivisible single item and agents with single-valued valuations for it, Myerson's characterization of the revenue-maximizing allocation does not carry over. Actually, similar simple characterizations seem out of reach in our case,\footnote{The approach of Cai et al. \cite{CDW12} might be used to obtain (an approximation of) the revenue-optimal mechanism as a solution of a mathematical program. However, this is far from Myerson's characterization in terms of simplicity.} but we would still like to have an estimate of the optimal expected revenue that can be extracted by the agents or a reasonable benchmark for it. Hence, we assume a {\em shortsighted behavior} for agents. Each agent responds to a pricing function with a demand for as large a fraction of the item as possible so that no smaller fraction yields her negative marginal utility. This allows us to reason about the revenue extracted by quite general mechanisms.

Our main goal is to provide bounds on the {\em revenue gap} between the highest expected revenue that can be achieved among given shortsighted agents and the highest expected revenue that can be achieved by a linear posted pricing mechanism. We do so by extending the concept of the {\em ex-ante relaxation} previously considered in \cite{A14,CHK07,AHNPY19}. Instead of comparing the revenue of the best linear pricing to the optimal expected revenue that can be extracted by shortsighted agents, we compare it to the best possible revenue of an ex-ante relaxation of the  problem, in which an {\em expected} fraction of at most $1$ is sold to them. We formulate the question regarding the revenue gap as a mathematical program that extends a mathematical program considered earlier for bounding the revenue gap of anonymous pricing in the indivisible item setting~\cite{AHNPY19}. Essentially, assuming $n$ shortsighted agents, our mathematical program has as variables the {\em selling probabilities} of the ``optimal'' mechanism and the probability distributions of the valuation marginals. The objective is to find an instance which yields the highest possible expected revenue under the constraints that the average total fraction of the item that is sold to the agents and the expected revenue of any linear pricing do not exceed $1$. Then, the maximum objective value of the mathematical program is an upper bound of the revenue gap.

Our technical contribution consists of bounds on the objective value of this mathematical program in three scenarios. Interestingly, these bounds depend on a natural parameter of the valuation functions, which we call {\em initial marginal over total value} (or IMOTV, for short), defined as the ratio of the marginal of the valuation function at point $0$ over its value at point $1$.

First, as a warm up, we consider the case of a single agent and bound the revenue gap by $\bigO(\ln{\kappa})$, where $\kappa$ denotes the maximum IMOTV of the valuation functions. This result holds even for non-regular distributions of the marginals of the valuations functions. A simple instance shows that it is tight, even assuming regular marginal distributions. 

Second, we bound the objective value of the mathematical program for the revenue gap by exploiting its relation to the mathematical program of~\cite{AHNPY19} via a black-box reduction. In this way, we obtain an $\bigO(\kappa^2)$ upper bound. This result is interesting for two reasons. First, it showcases that the maximum IMOTV of the valuation functions is the parameter of importance in the multi-agent case as well. Second, the proof applies to the more general setting of regular initial marginal distributions only; all valuation marginals at points different than $0$ can be non-regular.

Finally, our main result is an $\bigO(\ln{\kappa}+\ln{n})$ upper bound, where $n$ is the number of agents. The main idea of the proof is to identify an almost optimal feasible solution to the mathematical program, which has nice properties (e.g., for every agent, the marginal contribution of each item fraction to the expected revenue is upper bounded by a polynomial of $n$ and $\kappa$) that allow us to bound the maximum objective value, assisted by the constraints and additional technical statements for random variables. We remark that previous work on proving bounds on revenue gaps by bounding the objective value of mathematical programs~\cite{AHNPY19,JJLZ22,JLQTX19,JLTX19} relied heavily on extreme regular probability distributions for valuations that have a {\em triangle-shaped revenue-quantile curve}. Unfortunately, such tools cannot be used in our case since the distributions of the valuation marginals should correspond to a probability distribution over concave valuation functions. We bypass this technical obstacle by a simple geometric lemma in the proof of our third bound.

\subsection{Related work} 
Our paper falls within the ``simple vs.~optimal'' line of research in Bayesian mechanism design, initiated with the work of Harline and Roughgarden~\cite{HR09}, who studied the revenue gap of the second price auctions with an anonymous reserve. In the same thread, sequential posted pricing with an anonymous price has received much attention recently. Alaei et al.~\cite{AHNPY19} prove that this mechanism achieves a constant approximation of $2.72$ of the optimal expected revenue. They use a general strategy that we discuss in more detail later in Section~\ref{sec:indivisible}. The tight bound of $2.62$ follows by two papers by Jin et al.~\cite{JLQTX19, JLTX19}. We remark that all results for anonymous pricing carry over to our setting if the valuation functions are restricted to be linear. Indeed, the behavior of an agent with a linear valuation function against a linear pricing with a price of $p$ per unit is either to buy the whole item if her value for the whole item is higher than $p$ or refuse to buy otherwise. We mostly focus on non-linear valuation functions where divisibility differentiates our problem a lot.

A recent extension of single-item mechanisms is considered by Jin et al.~\cite{JJLZ22}, who study the setting with $k$ item copies available to be sold to potential buyers. However, a unit-demand assumption restricts the agent valuations to a single scalar value for the single item copy 
they aim to buy. Their main result is a tight revenue gap of $\Theta(\ln{k})$ under the unit-demand case. Instead, our modeling can be used to accommodate {\em multi-demand} agents with concave piecewise-linear valuation functions (with $k$ segments and, hence, of IMOTV at most $k$). Our third result implies a revenue gap of $O(\ln{k}+\ln{n})$ in this more general case.

A recent line of research that also deviates from the standard agent modeling assumptions focuses on agents with non-linear utility~\cite{FHL19,FHL23} that, e.g., captures budget constraints or risk aversion, and includes bounds on the revenue gap of anonymous pricing for such agents~\cite{FHL19}. However, these works usually assume single-parameter agents and a non-linear definition of their utility, while we focus on a standard quasi-linear definition of utility in a multi-parameter environment of a specific structure for the definition of valuations.

In the economics literature, divisible items have received attention, with the focus being mostly on whether bundling can be beneficial or not~\cite{P83} and on how to structure the sale as many auctions of shares~\cite{W79}. Perfect divisibility is considered in~\cite{LST15,SW11}. In another more related direction, the operations research community has considered resource allocation mechanisms to divide an item based on signals received by the agents, that are further used to impose payments to the agents. Among them, the proportional mechanism, first defined by Kelly~\cite{K97} and analyzed by Johari and Tsitsiklis~\cite{JT04} is the most popular one. Even though its social welfare has been analyzed extensively in stochastic settings that are very similar to ours~\cite{JT04,CV16,CST16} (see also~\cite{CV18} and the references therein), no revenue guarantees are known.
	
\subsection{Roadmap} The rest of the paper is structured as follows. We present our setup, introducing several notions and making connections to the relevant literature in Section~\ref{sec:prelim}. Our mathematical program for the revenue gap is also introduced there. In Section~\ref{sec:single-agent}, we present our tight revenue gap for a single agent. Section~\ref{sec:black-box} is devoted to the multi-agent case; our black-box reduction to the result of Alaei et al.~\cite{AHNPY19} is presented there. Our strongest logarithmic revenue gap is presented in Section~\ref{sec:direct}. We conclude with a short discussion on open problems in Section~\ref{sec:open}.

\section{Concepts, techniques, and notation}\label{sec:prelim}
We denote by $n$ the number of {\em agents} and use the integers in set $[n]=\{1,2, ..., n\}$ to identify them. A {\em valuation function} for agent $i$ is a monotone non-decreasing concave function $v_i:[0,1]\rightarrow \R_{\geq 0}$. We use $v'$ to refer to the (left) derivative of $v$.

Our main focus is on {\em linear posted pricing} mechanisms. Such a mechanism uses a scalar price $p$ and considers the agents in a predefined order. Whenever an agent is considered, she responds with her {\em demand}, i.e., the fraction of the item she is willing to buy at a price of $p$ per unit. The mechanism then allocates the agent a fraction of the item equal to either her demand or the fraction of the item that has not been previously allocated to other agents, whichever is smaller, and charges the agent an amount that is equal to $p$ multiplied with the fraction of the item the agent got. The demand $D(v,p)$ of an agent with valuation function $v$ at a price-per-unit $p$ is defined as the maximum fraction that maximizes the utility of the agent $D(v,p)\in \argmax_{x\in [0,1]}{\{v(x)-p\cdot x\}}$. Due to the concavity of the valuation functions, the demand is the maximum fraction $x\in [0,1]$ so that $v'(x)\geq p$. Hence, when the $n$ agents have a vector of valuation functions $\vv=\langle v_1, v_2, ..., v_n\rangle$, the revenue of the linear posted pricing mechanism that uses price-per-unit $p$ is $p\cdot \min\left\{1,\sum_{i\in [n]}{D(v_i,p)}\right\}$.

We will assume that all valuation functions $v$ have an {\em initial marginal over total value} (or IMOTV for short) ratio $v'(0)/v(1)$ bounded by a parameter $\kappa\geq 1$. We denote by $\calV_\kappa$ the set of all such valuation functions. 

We consider a Bayesian setting in which each agent $i\in [n]$ draws her valuation function $v_i$, independently from the other agents, from a publicly known probability distribution $\F_i$ over the valuation functions in $\calV_\kappa$. We write $v_i\sim \F_i$ and $\vv\sim \F$ to denote that the valuation function of agent $i$ and the vector of valuation functions $\vv=\langle v_1, v_2, ..., v_n\rangle$ are drawn randomly according to the probability distribution $\F_i$ and the joint probability distribution $\F$, respectively. 
For $i \in [n]$, let $v_i\sim \F_i$. For $x \in [0,1]$, we denote by $F_{i,x}$ the cumulative density function of the probability distribution of the derivative $v_i'(x)$ of the valuation function $v_i$ at point $x$. We will often restrict the probability distribution $\F$ so that $F_{i,x}$ is regular for $i\in[n]$ and $x\in[0,1]$. 


Our modeling assumption includes, as a special case, the model of \cite{SW11}, where the valuation function $v_i(z)=t_i\cdot h_i(z)$ has a scalar part $t_i$ and a function part $h_i(z)$. The function part is known to the seller and can be used by the mechanism. The scalar part is drawn from a probability distribution, in which regularity can be imposed. This is essentially a {\em single-parameter} environment, where it is a simple exercise to apply Myerson's approach~\cite{M81} of maximizing revenue by maximizing virtual welfare. This means that revenue-optimal mechanisms have a well-known structure in this setting. In contrast, our setting is {\em multi-parameter}, and our modeling allows for richer distributions over valuation functions. Unfortunately, Myerson’s machinery cannot be applied anymore, so we do not have any clean form of the revenue-maximizing mechanism. This makes even the definition of the revenue gap in our case, as well as the proofs of bounds on its value, more challenging.

\subsection{Selling an indivisible item with an anonymous price}\label{sec:indivisible}
We will adapt the approach of \cite{AHNPY19} for selling an indivisible item using sequential posted pricing and an {\em anonymous price}. In that classic setting, each agent $i\in [n]$ has a scalar valuation $v_i$ drawn from a probability distribution with regular cumulative density function $H_i$. 

Alaei et al.~\cite{AHNPY19} use the mathematical program (\ref{eq:mp-indivisible}) below to bound the gap between the optimal expected revenue and the maximum expected revenue that can be achieved with an anonymous price. Actually, instead of comparing directly to the optimal revenue (in fact, this is done in follow-up work by Jin et al.~\cite{JLQTX19}), they compare to the best possible revenue of the {\em ex-ante relaxation}, i.e., the maximum expected revenue that can be achieved by a mechanism that sells the item to at most one agent {\em in expectation}. Their mathematical program uses as variables (1) the regular cumulative density function $H_i$ of the valuation of agent $i$ for the item and (2) the probability $r_i$ that the item is given to agent $i$ in the revenue-maximizing allocation for the ex-ante relaxation. 
\begin{align}\label{eq:mp-indivisible}
\mbox{maximize} &\sum_{i\in [n]}{r_i \cdot H^{-1}_i(1-r_i)}\\\nonumber
\mbox{subject to} & \sum_{i\in [n]}{r_i}\leq 1\\\nonumber
& p\cdot \left(1-\prod_{i\in [n]}{H_i(p)}\right) \leq R, \forall p>R\\\nonumber
& r_i \in [0,1], \forall i\in [n]\\\nonumber
& H_i \mbox{ is a regular cdf}, \forall i\in [n]
\end{align}
Viewed together with the first constraint, the objective of the mathematical program (\ref{eq:mp-indivisible}) is to maximize the revenue of the ex-ante relaxation. The second set of constraints implements the restriction that no anonymous pricing $p$ yields a revenue higher than $R$. The ratio between the objective value of (\ref{eq:mp-indivisible}) and $R$ is an upper bound to the revenue gap. The main result of \cite{AHNPY19} is as follows. 
\begin{theorem}[Alaei et al.~\cite{AHNPY19}]\label{thm:rev-Alaei-et-al}
	For every $R\geq 1$, the objective value of the mathematical program (\ref{eq:mp-indivisible}) is at most $e\cdot R$.
\end{theorem}
Our second bound on the revenue gap when selling a divisible item with linear pricing will use the result of \cite{AHNPY19} as a black box. We remark that the original result in \cite{AHNPY19} uses specifically $R=1$. The extension we consider here is without loss of generality and is used to simplify our exposition in Section~\ref{sec:black-box}. 

\subsection{Shortsighted agents}\label{sec:benchmark}
Let us now return to our setting. As a benchmark for the comparison and assessment of the revenue of linear pricing, we consider general pricing schemes that may discriminate between different agents using different {\em pricing functions} of the form $p:[0,1]\rightarrow \R_{\geq 0}$, where $p(x)$ denotes the price an agent is asked to pay for getting an item fraction of $x$. We extend the notion of the demand by slightly abusing notation as follows. We say that, given a pricing function $p$ with left derivative $p'$, the demand $D(v,p)$ of an agent with a valuation function $v$ is the maximum item fraction $x$ so that $v'(t)\geq p'(t)$ for every $t\in [0,x]$. Thus, we assume that agents are {\em shortsighted} and aim to only locally maximize their utility. 

Let us explain the notion of shortsightness with an example. Consider an agent with the valuation function
\begin{align*}    
v(x) &=\begin{cases}
    12x & x\in [0,1/3]\\
    1+9x & x\in [1/3,2/3]\\
    5+3x & x\in [2/3,1]
    \end{cases}
\end{align*}
and the pricing function
\begin{align*}
    p(x) &= \begin{cases}
        10x & x\in [0,1/2]\\
        3+4x & x\in [1/2,1]
    \end{cases}.
\end{align*}
The utility of the agent as a function of the item fraction $x$ is now
\begin{align*}
    \text{utility}(x) &=\begin{cases}
        2x & x\in [0,1/3]\\
        1-x & x\in [1/3,1/2]\\
        -2+5x & x\in [1/2,2/3]\\
        2-x & x\in [2/3,1]
\end{cases}.
\end{align*}
Hence, it increases as $x$ ranges from point $0$ to point $1/3$ up to the value of $2/3$, then decreases and gets the value of $1/2$ at point $1/2$, then increases again getting the highest value of $4/3$ at point $2/3$, and decreases afterwards. According to our assumption of a shortsighted agent, her demand is $D(v,p)=1/3$ (i.e., the first local maximum of the utility) even though the utility is globally maximized later at point $2/3$.

Our objective is to prove bounds on the {\em revenue gap} of linear posted pricing, i.e., on the maximum ratio, over all possible instances that follow our setting, between the expected revenue that can be achieved with shortsighted agents and the maximum revenue that can be achieved with linear posted pricing. We remark that, due to the concavity of the valuation functions, shortsighted agents actually respond as utility-maximizers to linear pricing.

\subsection{A mathematical program for the revenue gap}\label{subsec:rev:mp}
Extending the approach of Alaei et al.~\cite{AHNPY19}, we also use an ex-ante relaxation to upper-bound the maximum revenue that can be achieved with shortsighted agents. We consider general pricing schemes consisting of a pricing function $p_i:[0,1]\rightarrow \R^+$ for each agent $i\in [n]$. For every agent $i\in [n]$ and every $x\in [0,1]$, let $q_i(x)$ be the probability that a fraction of at least $x$ is bought by agent $i$ with pricing $p_i$ (i.e., $D(v_i,p_i)\geq x$). For $i\in [n]$ and $x\in [0,1]$, observe that the shortsighted agent $i$ will have a demand $D(v_i,p_i)$ that is at least $x$ only if $v'_i(x)\geq p'_i(x)$. Hence, $q_i(x)\leq\Pr[v'_i(x)\geq p'_i(x)]=1-F_{i,x}(p'_i(x))$ and, equivalently, $p'_i(x)\leq F_{i,x}^{-1}(1-q_i(x))$. Using the quantities $q_i(x)$ and $p'_i(x)$, we can upper-bound the expected payment by agent $i$ as $\int_0^1{q_i(x) \cdot p'_i(x)\ud x}\leq \int_0^1{q_i(x)\cdot F_{i,x}^{-1}(1-q_i(x))\ud x}$. Intuitively, $q_i(x)$ denotes the probability that agent $i$ buys her $x$-th point of the item, $p'_i(x)$ denotes the payment increase due to this point, and the quantity $q_i(x)\cdot F_{i,x}^{-1}(1-q_i(x))$ represents (an upper bound of) the marginal contribution of agent $i$ and point $x$ to the expected revenue. Hence, the maximum expected revenue that can be achieved by shortsighted agents is upper-bounded by the quantity
\begin{align}\label{eq:optimal-revenue}
\sum_{i\in [n]}{\int_0^1{q_i(x)F_{i,x}^{-1}(1-q_i(x))\ud x}},
\end{align}
under the constraint that {\em on average}, no more than the whole item is available for purchase by the agents. 
We will refer to expression (\ref{eq:optimal-revenue}) as the expected revenue of the ex-ante relaxation.

Next, we extend the approach of~\cite{AHNPY19} to relate the ex-ante revenue relaxation with shortsighted agents to the maximum revenue that can be achieved with linear pricing. To determine the revenue gap, we can further require that any linear pricing with price per unit $p$ has revenue at most $1$ and ask: ``How large can the expected revenue of the ex-ante relaxation be?'' The answer to this question is given by the  mathematical program (\ref{eq:mp-original}) below. The objective is to maximize expression (\ref{eq:optimal-revenue}), which upper-bounds the ex-ante relaxation of the revenue that can be achieved with shortsighted agents. Constraint (\ref{eq:constraint-bounded-q}) ensures that the average fraction allocated to the agents does not exceed the available item. The LHS of (\ref{eq:constraint-min}) is just the expected revenue of linear pricing with a price $p$ per unit. The mathematical program has as variables the probability distributions over the valuation functions of $\calV_\kappa$, which are further described by cdf's $F_{i,x}$ for $i\in [n]$ and $x\in [0,1]$, and the probabilities $q_i(x)$ for $i\in [n]$ and $x\in [0,1]$; by its intending meaning as the probability that agent $i$ gets a fraction of at least $x$, $q_i$ (called {\em selling probability} from now on) is required to be monotone non-increasing in $x$. 
\begin{align}\label{eq:mp-original}
\mbox{maximize} &\sum_{i\in [n]}{\int_0^1{q_i(x)\cdot F^{-1}_{i,x}(1-q_i(x))}\ud x}\\\label{eq:constraint-bounded-q}
\mbox{subject to} & \sum_{i\in [n]}{\int_0^1{q_i(x)}\ud x}\leq 1\\\label{eq:constraint-min}
& p\cdot\E_{\vv\sim \F}\left[\min\left\{1,\sum_{i\in [n]}{D(v_i,p)}\right\}\right]\leq 1, \forall p>1\\\nonumber
& \mbox{monotone non-increasing } q_i: [0,1]\rightarrow [0,1], \forall i\in [n], x\in [0,1]\\\nonumber
& \supp(\F_i)\subseteq \calV_\kappa, F_{i,x} \mbox{ is a regular cdf}, \forall i\in [n], x\in [0,1]
\end{align}

In the following, we frequently abbreviate the quantity $q_i(x)\cdot F_{i,x}^{-1}(1-q_i(x))$ as $R_{q_i,\F_i}(x)$ for $i\in [n]$ and $x\in [0,1]$.

\section{Warming up: the single-agent case}\label{sec:single-agent}
We begin by presenting a tight bound on the revenue gap for the case of a single agent. For simplicity of exposition, we drop the index $i$ from notation. Actually, we use a simplified version of the mathematical program (\ref{eq:mp-original}) where we require neither the selling probability function $q(x)$ to be monotone non-increasing nor the probability distributions $F_x$ to be regular. Since we have just a single agent, the constraint (\ref{eq:constraint-bounded-q}) is now redundant. In particular, we will upper-bound the objective value of the following simplified mathematical program:
\begin{align}\label{eq:mp-single-agent}
\mbox{maximize } & \int_0^1{q(x)\cdot F^{-1}_x(1-q(x))}\ud x\\\nonumber
\mbox{subject to } & \,p\cdot \E_{\vv\sim \F}\left[D(v,p)\right]\leq 1, \forall 
p>1\\\nonumber
& q(x)\in [0,1], \forall x\in [0,1]\\\nonumber
& \supp(\F)\subseteq \calV_{\kappa}, F_x: [0,1]\rightarrow [0,1] \mbox{ is a cdf}, \forall x\in [0,1]
\end{align}

To do so, we will need three technical lemmas. The first one follows easily by the concavity of the valuation functions and the optimality of selling probabilities.

\begin{lemma}\label{lem:mon-single-agent}
    Let $(\q,\F)$ be an optimal solution of the mathematical program (\ref{eq:mp-single-agent}). Then, $R_{q,\F}(x)$ is non-increasing in $x$.
\end{lemma}

\begin{proof}
    Let $0\leq x_1<x_2\leq 1$. Due to the concavity of the valuation functions, $F_x^{-1}(t)$ is monotone non-increasing in $x$ for every $t\in [0,1]$. Thus,
    \begin{align*}
        R_{q,\F}(x_2) &=q(x_2)\cdot F_{x_2}^{-1}(1-q(x_2)) \leq q(x_2)\cdot F_{x_1}^{-1}(1-q(x_2))\\
        &\leq q(x_1)\cdot F_{x_1}^{-1}(1-q(x_1)) = R_{q,\F}(x_1)
    \end{align*}
as desired. The second inequality follows since $q(x_1)$ is the optimal selling probability at point $x_1$ (under no constraints).
\end{proof}

The next lemma is due to the first set of constraints of the mathematical program (\ref{eq:mp-single-agent}).
\begin{lemma}\label{lem:lp-constraint-single-agent}
    Let $(\q,\F)$ be a solution of the  mathematical program (\ref{eq:mp-single-agent}). Then, $R_{q,\F}(x)\leq 1/x$ for $x\in (0,1]$. 
\end{lemma}

\begin{proof}
    Notice that, for any $x\in [0,1]$, by setting the linear price $p\cdot x$ with $p=F_x^{-1}(1-q(x))$, the derivative of the valuation $v$ drawn from distribution $\F$ satisfies $v'(x)\geq p$ with probability at least $q(x)$. This yields an expected demand of $\E_{v\sim \F}\left[D(v,F_x^{-1}(1-q(x)))\right]\geq x\cdot q(x)$. The lemma follows since, by the constraint of the mathematical program (\ref{eq:mp-single-agent}) with $p=F_x^{-1}(1-q(x))$, we get $\E_{v\sim \F}\left[D(v,F_x^{-1}(1-q(x)))\right]\leq 1/F_x^{-1}(1-q(x))$. 
\end{proof}

The next lemma indicates that the function $\int_0^t{R_{q,\F}(x)\ud x}$, which was essentially proved to be concave in Lemma~\ref{lem:mon-single-agent}, has an IMOTV bounded by a polynomial in $\kappa$.
\begin{lemma}\label{lem:imotv-single-agent}
    Let $(\q,\F)$ be a solution of the  mathematical program (\ref{eq:mp-single-agent}). Then, it holds 
    \begin{align}
        R_{q,\F}(0)\leq 4\kappa^2\int_0^1{R_{q,\F}(x)\ud x}.
    \end{align}
\end{lemma}

\begin{proof}
Denote by $V\subseteq \calV_\kappa$ the set of valuation functions that are drawn with non-zero probability and let $S=\left\{v\in V: v'(0)\geq F_0^{-1}(1-q(0))\right\}$ be the set of valuation functions with derivative at least $F_0^{-1}(1-q(0))$ at point $0$. Since $q(x)$ is the optimal selling probability at point $x\in [0,1]$, we get 
    \begin{align}\label{eq:step1-single-agent}
        R_{q,\F}(x) &=q(x)\cdot F_x^{-1}(1-q(x))\geq q(0)\cdot F_x^{-1}(1-q(0)).
    \end{align}
    By the definition of set $S$, we have $\Pr_{v\sim \F}[v\in S]=1-q(0)$. Hence, 
    \begin{align}\label{eq:step2-single-agent}
        F_x^{-1}(1-q(0)) & \geq \min_{v\in S}{v'(x)}.
    \end{align}
    By Equations (\ref{eq:step1-single-agent}) and (\ref{eq:step2-single-agent}), we get
    \begin{align}\label{eq:relate-min-single-agent}
        \int_0^1{\min_{v\in S}{v'(x)}\ud x} &\leq \frac{1}{q(0)}\cdot \int_0^1{R_{q,\F}(x)\ud x}
    \end{align}
    Now, let $t\in [0,1]$ and $\widehat{v}$ a valuation function of $S$ satisfying $\widehat{v}'(t)=\min_{v\in S}{v'(t)}$. We have
    \begin{align}\nonumber
        \int_0^1{\widehat{v}'(x)\ud x} &= \int_0^t{\widehat{v}'(x)\ud x}+\int_t^1{\widehat{v}'(x)\ud x}\\\nonumber
        &\leq t\cdot \widehat{v}'(0)+(1-t)\cdot \widehat{v}'(t)\\\nonumber
        &\leq t\cdot \widehat{v}'(0)+\frac{1-t}{t}\cdot \int_0^t{\min_{v\in S}{v'(x)}\ud x}\\\label{eq:compare-imotv-step3-single-agent}
        &\leq t\cdot \widehat{v}'(0)+ \frac{1-t}{t\cdot q(0)}\cdot \int_0^1{R_{q,\F}(x)\ud x}.
    \end{align}
    The first two inequalities follow due to the concavity of the valuation function $\widehat{v}$ and its definition, while the third one is due to Equation (\ref{eq:relate-min-single-agent}). Now, we divide both sides of Equation (\ref{eq:compare-imotv-step3-single-agent}) by $\widehat{v}'(0)$ and use the facts that the valuation function $\widehat{v}$ has IMOTV at most $\kappa$ and belongs to set $S$ (implying that $\widehat{v}'(0)\geq F_0^{-1}(1-q(0))$). We obtain
    \begin{align}\label{eq:compare-imotv-step4-single-agent}
        \frac{1}{\kappa} &\leq \frac{\int_0^1{\widehat{v}'(x)\ud x}}{\widehat{v}'(0)} \leq t +\frac{1-t}{t}\cdot \frac{\int_0^1{R_{q,\F}(x)\ud x}}{q(0)\cdot \widehat{v}'(0)} \leq t+\frac{1-t}{t}\cdot \frac{\int_0^1{R_{q,\F}(x)\ud x}}{R_{q,\F}(0)}.
    \end{align}
    Rearranging Equation (\ref{eq:compare-imotv-step4-single-agent}), we get
    \begin{align}\label{eq:compare-imotv-step5-single-agent}
        \frac{R_{q,\F}(0)}{\int_0^1{R_{q,\F}(x)\ud x}} &\leq \frac{1-t}{t\cdot \left(\frac{1}{\kappa}-t\right)},
    \end{align}
    and setting $t=1-\sqrt{1-\frac{1}{\kappa}}$, Equation (\ref{eq:compare-imotv-step5-single-agent}) yields
    \begin{align*}
        \frac{R_{q,\F}(0)}{\int_0^1{R_{q,\F}(x)\ud x}} &\leq \left(1-\sqrt{1-\frac{1}{\kappa}}\right)^{-2}\leq 4\kappa^2,
    \end{align*}
    as desired. The last inequality follows since the function $1-\sqrt{1-z}$ is convex in $[0,1]$ and has a derivative equal to $1/2$ at $z=0$. Hence, $1-\sqrt{1-z}\geq z/2$ and the inequality follows using $z=1/\kappa$.
\end{proof}

We are ready to prove our first result using Lemmas~\ref{lem:mon-single-agent}, \ref{lem:lp-constraint-single-agent}, and \ref{lem:imotv-single-agent}; notice that the derivatives of the valuation functions may follow non-regular probability distributions.
\begin{theorem}\label{thm:single-agent}
    The revenue gap of linear posted pricing for a single agent with random concave valuations of IMOTV at most $\kappa$ is at most $\bigO(\ln{\kappa})$.
\end{theorem}

\begin{proof}
Let $t\in [0,1]$. The objective of the mathematical program (\ref{eq:mp-single-agent}) is
\begin{align*}
    \int_0^1{R_{q,\F}(x)\ud x} &= \int_0^t{R_{q,\F}(x)\ud x} + \int_t^1{R_{q,\F}(x)\ud x} \\
    &\leq t\cdot R_{q,\F}(0)+\int_t^1{\frac{\ud x}{x}}\\
    &\leq 4t\kappa^2\cdot \int_0^1{R_{q,\F}(x)\ud x}+\ln{\frac{1}{t}}.
\end{align*}
The first inequality follows by Lemmas~\ref{lem:mon-single-agent} and~\ref{lem:lp-constraint-single-agent} and the second one by Lemma~\ref{lem:imotv-single-agent}. Equivalently, we have 
\begin{align*}
    \int_0^1{R_{q,\F}(x)\ud x} &\leq \frac{\ln{\frac{1}{t}}}{1-4t\kappa^2}.
\end{align*}
Using $t=\frac{1}{8\kappa^2}$, we get that the objective of the mathematical program (\ref{eq:mp-single-agent}) is upper-bounded by $2\ln{(8\kappa^2)} \in \bigO(\ln{\kappa})$.
\end{proof}

It is not hard to extend the result to the case of $n$ agents (by just considering the agent who contributes the most to the optimal revenue); again, the derivatives of the valuation functions may follow non-regular probability distributions.
\begin{cor}
    The revenue gap of linear posted pricing for $n$ agents with random concave valuations with IMOTV at most $\kappa$ is at most $\bigO(n\ln{\kappa})$.
\end{cor}

\noindent We remark that the linear dependency on $n$ is necessary, as indicated by a lower bound of~\cite{AHNPY19} on the revenue gap of anonymous pricing for agents with non-regular valuations.

\subsection{A lower bound}
We now show that our bound in Theorem~\ref{thm:single-agent} is tight, even for valuation function with derivative values following regular probability distributions.

\begin{theorem}
For every $\kappa> 1$, there exists an instance with a single agent with random concave valuation functions of IMOTV at most $\kappa$ and with all valuation derivative values following a regular probability distribution, so that the revenue gap of linear pricing is $\Omega(\ln{\kappa})$. 
\end{theorem}

\begin{proof}
Let $\rho$ be such that $1+\ln{(\kappa \cdot\rho)}=\rho$, i.e., $\rho\geq 1+\ln{\kappa}$. For a parameter $t\geq 0$, let $v_t$ be the valuation function defined as 
\begin{align*}
	v_t(x) = 
	\begin{cases}
	(\kappa+t) \cdot x, & x\in \left[0, \frac{1}{\kappa \cdot \rho}\right]\\
	\frac{1}{\rho}\cdot \left(1+\ln{(\kappa\cdot \rho\cdot x)}\right)+t\cdot x, & x\in \left[\frac{1}{\kappa \cdot \rho},1\right]
	\end{cases}
	\end{align*}
The function $v_t(x)$ has derivative (with respect to $x$) equal to $\kappa+t$ for $x\in \left[0,\frac{1}{\kappa\cdot \rho}\right]$ and to $\frac{1}{\rho\cdot x}+t$ for
$x\in \left[\frac{1}{\kappa\cdot \rho},1\right]$. Thus, it is concave. Since $t\geq 0$, its IMOTV is $\frac{v'_t(0)}{v_t(1)}=\frac{\kappa+t}{\frac{1}{\rho}\cdot \left(1+\ln{(\kappa\cdot \rho}\right)+t} = \frac{\kappa+t}{1+t}\leq \kappa$.

Consider a single agent who draws a random valuation function $v_t$ by selecting $t$ uniformly at random from the interval $[0,1/\rho]$.\footnote{We remark that we can simplify the proof by setting $t=0$. The particular choice of a uniform $t\in[0,1/\rho]$ results in regular valuation derivatives.} Thus, the derivative $v'_t(x)$ is uniformly random in $[\kappa,\kappa+1/\rho]$ if $x\in \left[0,\frac{1}{\kappa\cdot \rho}\right]$ and uniformly random in $\left[\frac{1}{\rho\cdot x},\frac{1}{\rho\cdot x}+\frac{1}{\rho}\right]$ if $x\in \left[\frac{1}{\kappa\cdot \rho},1\right]$, i.e., drawn from regular probability distributions.

Using the pricing function $p(x)=v_0(x)$, the agent will buy the whole item, yielding a revenue of $v_0(1)=1$. Now, consider the linear pricing $p$ and first observe that $D(v_t,p)\leq D(v_{1/\rho},p)$ for every $t\in [0,1/\rho]$. Next, notice that $D(v_{1/\rho},p)=0$ if $p>\kappa+1/\rho$, $D(v_{1/\rho},p)=\frac{1}{\rho\cdot (p-1/\rho)}$ if $2/\rho \leq p \leq \kappa+1/\rho$, and $D(v_{1/\rho},p)=1$ if $p<2/\rho$. Thus, the upper bound of $p\cdot D(v_{1/\rho},p)$ on the expected revenue with the linear pricing $p$ is $0$, at most $\frac{p}{p\cdot (p-1/\rho)}\leq \frac{2}{\rho}$, and less than $2/\rho$, respectively. In this way, we obtain a revenue gap of at least $\rho/2\in \Omega(\ln{\kappa})$.
\end{proof}

\section{The multi-agent case: a black-box reduction}\label{sec:black-box}

We now prove an upper bound on the objective value of the mathematical program (\ref{eq:mp-original}) by relating it to the mathematical program (\ref{eq:mp-indivisible}), and exploiting the revenue gap of \cite{AHNPY19} for anonymous item pricing (Theorem~\ref{thm:rev-Alaei-et-al}). In particular, we prove the following statement for multiple agents.

\begin{theorem}\label{thm:rev-main}
The revenue gap when selling a divisible item using linear pricing to multiple agents with random concave valuations of IMOTV at most $\kappa$ and with derivatives following regular probability distributions is at most $\bigO(\kappa^2)$.
\end{theorem}

\begin{proof}
Our main tool is the following transformation. We use the notation $(\q,\F)$ as abbreviation for the probabilities $q_i(x)$ for every agent $i\in [n]$ and $x\in [0,1]$ and the cumulative density functions $F_{i,x}(t)$ for every agent $i\in [n]$, $x\in [0,1]$, and $t\geq 0$. Given the pair $(\q,\F)$, we define $r_i=\int_0^1{q_i(x)\ud{x}}$ and $H_i(t) = F_{i,0}(2\kappa t)$ for every agent $i\in [n]$ and $t\geq 0$, and use the pair $(\rr,\HH)$ as their abbreviation. 

Bounding the objective value of mathematical program (\ref{eq:mp-original}) has two steps, implemented in Lemmas~\ref{lem:feasible} and \ref{lem:mp2-vs-mp1}, respectively.

\begin{lemma}\label{lem:feasible}
Given a feasible solution $(\q,\F)$ of the mathematical program (\ref{eq:mp-original}), the solution $(\rr,\HH)$ is a feasible solution of the mathematical program (\ref{eq:mp-indivisible}) with $R=2\kappa-1$.
\end{lemma}

\begin{proof}
Clearly, $r_i\in [0,1]$ and $H_i$ is a regular cumulative density function (since $F_{0,i}$ is regular as well). Furthermore, the solution $(\rr,\HH)$ satisfies the first constraint of the mathematical program (\ref{eq:mp-indivisible}) since $(\q,\F)$ satisfies constraint (\ref{eq:constraint-bounded-q}) of the mathematical program (\ref{eq:mp-original}). In the following, we show that $(\rr,\HH)$ satisfies the second constraint of the mathematical program (\ref{eq:mp-indivisible}) as well. We will need three technical lemmas.

\begin{lemma}\label{lem:min}
Let $X_1$, $X_2$, ..., $X_k$ be independent random variables with $X_i\in [0,1]$ for $i\in [k]$ and let $X=\sum_{i\in [k]}{X_i}$. Then, $\E[\min\{1,X\}] \geq 1-\prod_{i\in [k]}{(1-\E[X_i])}$.
\end{lemma} 

\begin{proof} We claim that $\E[\min\{1,X\}]$ is minimized when $X_1$, $X_2$, ..., $X_k$ are Bernoulli random variables. Then, it will be
	\begin{align*}
	\E[\min\{1,X\}] &= 1-\prod_{i\in [k]}{\Pr[X_i=0]}=1-\prod_{i\in [k]}{(1-\E[X_i])},
	\end{align*} 
	and the lemma will follow.
	
	To prove this claim, we show that by replacing the random variable $X_k$ with the Bernoulli random variable $Y_k$ with $\Pr[Y_k=1]=\E[X_k]$ and $\Pr[Y_k=0]=1-\E[X_k]$ the expectation $\E[\min\{1,X\}]$ can only become smaller. The claim will follow by repeating this argument and replacing each random variable $X_i$ with a Bernoulli one that has the same expectation. 
	
	Formally, let $Y=X_1+...+X_{k-1}+Y_k=X'+Y_k$; we will show that $\E[\min\{1,X\}]\geq \E[\min\{1,Y\}]$. Denoting by $G$ the cdf of the random variable $Y_k$, we have that, conditioned on $X'=w$, the expected contribution of $X_k$ to $\min\{1,X\}$ is 
	\begin{align*}
	\E[\min\{1,X\}-X'|X'=w] &= \int_0^{1-w}{(1-G(z))\ud z} \geq (1-w)\int_0^1{(1-G(z))\ud z}\\
	&=(1-w)\E[X_k] =(1-w)\Pr[Y_k=1]=\E[\min\{1,Y\}-X'|X'=w].
	\end{align*}
	The inequality follows since $1-G(z)$ is non-increasing in $z$ and, subsequently, $\int_0^t{\left(1-G(z)\right)\ud z}$ is concave in $t$ and has no point below the line $t\int_0^t{\left(1-G(z)\right)\ud z}$ for $t\in [0,1]$. Denoting the pdf of the random variable $X'$ by $f$, we have
	\begin{align*}
	\E[\min\{1,X\}] &= \int_0^1{f(w)(w+\E[\min\{1,X\}-X'|X'=w])\ud w}\\
	&\geq \int_0^1{f(w)(w+\E[\min\{1,Y\}-X'|X'=w])\ud w}\\
	&= \E[\min\{1,Y\}],
	\end{align*}
	as desired.
\end{proof}

\begin{lemma}\label{lem:second-constraint}
For every agent $i\in [n]$ and any price $p$, it holds that $\E_{v_i\sim \F_i}[D(v_i,p)] \geq \frac{1-H_i(p)}{2\kappa-1}$.
\end{lemma}

\begin{proof}
Assume that $v'_i(0)\geq 2\kappa p$; hence, $v_i(1)\geq 2p$. By the definition of $D(v_i,p)$, we have that $v'_i(z)\leq p$ for $z\geq D(v_i,p)$. Hence, 
\begin{align}\label{eq:y-star-1}
v_i(D(v_i,p)) &\geq v_i(1)-v'_i(D(v_i,p))\cdot (1-D(v_i,p)) \geq v_i(1)-p\cdot (1-D(v_i,p))
\end{align}
Furthermore, by our assumption on the IMOTV bound of the valuation functions, we have $v'_i(z)\leq \kappa \cdot v_i(1)$ for $z\leq D(v_i,p)$. Hence, 
\begin{align}\label{eq:y-star-2}
v_i(D(v_i,p)) &\leq \kappa \cdot v_i(1) \cdot D(v_i,p).
\end{align}
Now, (\ref{eq:y-star-1}) and (\ref{eq:y-star-2}) yield
\begin{align*}
D(v_i,p) & \geq \frac{v_i(1)-p}{\kappa \cdot v_i(1)-p} \geq \frac{1}{2\kappa-1}.
\end{align*}
Thus,
\begin{align*}
\E_{v_i\sim \F_i}[D(v_i,p)] &\geq \E_{v_i\sim \F_i}[D(v_i,p) \one\{v'_i(0)\geq 2\kappa p\}] \geq \frac{1-F_{i,0}(2\kappa p)}{2\kappa-1} = \frac{1-H_i(p)}{2\kappa-1}. \qedhere
\end{align*}
\end{proof}

\begin{lemma}\label{lem:prod}
Let $k$ be an integer and $0\leq z_1, z_2, ..., z_k\leq 1$. Then, for every $t\in (0,1)$, it holds that
\begin{align*}
1-\prod_{i\in [k]}{(1-t\cdot z_i)} &\geq t \left(1-\prod_{i\in [k]}{(1-z_i)}\right)
\end{align*}
\end{lemma}

\begin{proof}
	It suffices to show that the LHS is concave as a function of $t$; then it is at least as high as the line that connects points $(0,0)$ and $(1,1-\prod_{i\in [k]}{(1-z_i)})$, i.e., the RHS of the above inequality.  Indeed, the first derivative of the LHS is equal to 
	\begin{align*}
	\prod_{i\in [k]}{(1-t\cdot z_i)} \cdot \sum_{i\in [k]}{\frac{z_i}{1-t\cdot z_i}}
	\end{align*} 
	and its second derivative is equal to
	\begin{align*}
	-\prod_{i\in [k]}{(1-t\cdot z_i)} \left(\left(\sum_{i\in [k]}{\frac{z_i}{1-t\cdot z_i}}\right)^2-\sum_{i\in [k]}{\left(\frac{z_i}{1-t\cdot z_i}\right)^2}\right) &<0,
	\end{align*}
	as desired.
\end{proof}

Now, using the constraint (\ref{eq:constraint-min}) of the mathematical program (\ref{eq:mp-original}) for $p>2\kappa-1$ and applying Lemma~\ref{lem:min} (with $k=n$ and $X_i=D(v_i,p)$ for $i\in [n]$), Lemma~\ref{lem:second-constraint} and Lemma~\ref{lem:prod}, we have
\begin{align*}
1&\geq p\cdot \E_{\vv\sim \F}{\left[\min\left\{1,\sum_{i\in [n]}{D(v_i,p)}\right\}\right]} \geq p \cdot \left(1-\prod_{i\in [n]}{\left(1-\E_{v_i\sim \F_i}[D(v_i,p)]\right)}\right)\\
&\geq p \cdot \left(1-\prod_{i\in [n]}{\left(1-\frac{1-H_i(p)}{2\kappa-1}\right)}\right) \geq \frac{p}{2\kappa-1} \cdot \left(1-\prod_{i\in [n]}{H_i(p)}\right),
\end{align*}
i.e., equivalently, 
\begin{align*}
p\left(1-\prod_{i\in [n]}{H_i(p)}\right) &\leq R, \forall p>R,
\end{align*}
with $R=2\kappa-1$, as the second constraint of the mathematical program (\ref{eq:mp-indivisible}) requires.
\end{proof}

The second step of our proof is to relate the objective values of the two mathematical programs (\ref{eq:mp-original}) and (\ref{eq:mp-indivisible}).

\begin{lemma}\label{lem:mp2-vs-mp1}
The objective value of the mathematical program (\ref{eq:mp-original}) at solution $(\q,\F)$ is at most $2\kappa$ times the objective value of the mathematical program (\ref{eq:mp-indivisible}) with $R=2\kappa-1$ at solution $(\rr,\HH)$.
\end{lemma}

\begin{proof}
	Notice that, by the definition of $\HH$, we have $F^{-1}_{i,0}(t)=2\kappa H^{-1}_i(t)$. Hence, 
	\begin{align*}
	\int_0^1{q_i(x) F^{-1}_{i,x}(1-q_i(x))\ud{x}} &\leq \int_0^1{q_i(x) F^{-1}_{i,0}(1-q_i(x))\ud{x}}\\
	&\leq \int_0^1{q_i(x)\ud{x}} \cdot F^{-1}_{i,0}\left(1-\int_0^1{q_i(x)\ud{x}}\right)\\
	&= 2\kappa r_i H^{-1}_i(1-r_i).
	\end{align*}
	The first inequality follows by the concavity of valuations which implies $F^{-1}_{i,x}(t)$ is non-increasing in $x$. The second inequality follows by applying Jensen inequality; recall that, by the regularity of $v'_i(0)$, the function $q \cdot F^{-1}_{i,0}(1-q)$ is concave in $q$.
\end{proof}

Theorem~\ref{thm:rev-main} now follows by Theorem \ref{thm:rev-Alaei-et-al} and Lemmas~\ref{lem:feasible} and~\ref{lem:mp2-vs-mp1}. Indeed, these three statements imply a revenue gap of at most $2\kappa(2\kappa-1)e$.
\end{proof}

\section{The multi-agent case: a direct approach}\label{sec:direct}
We now present a direct approach to bound the objective value of the mathematical program (\ref{eq:mp-original}), which will give us a better (logarithmic) dependency on $\kappa$ at the expense of a logarithmic dependency on the number of agents. In particular, we prove the following statement for multiple agents.

\begin{theorem}\label{thm:rev-main-direct}
    The revenue gap when selling a divisible item using linear pricing to $n$ agents with random concave valuations of IMOTV at most $\kappa$ and with derivatives that follow regular probability distributions is at most $\bigO(\ln{\kappa}+\ln{n})$.
\end{theorem}

\begin{proof}
The proof of Theorem~\ref{thm:rev-main-direct} is decomposed into two main parts. First, we show that the mathematical program (\ref{eq:mp-original}) has a feasible solution that satisfies certain properties. Then, we use this solution and its properties to bound the maximum objective value.

\begin{lemma}\label{lem:relaxed-constraint}
    Any feasible solution $(\q,\F)$ of the mathematical program (\ref{eq:mp-original}) satisfies 
    \begin{align*}
        \sum_{i\in [n]}{\int_0^1{(1-F_{i,x}(p))\ud x}}\leq \frac{1}{p-1}, \forall p>1.
    \end{align*}
\end{lemma}

\begin{proof}
We will prove the statement using the constraint (\ref{eq:constraint-min}) of the mathematical program (\ref{eq:mp-original}) and the next technical lemma.

\begin{lemma}\label{lem:min-new}
Let $X_1$, $X_2$, ..., $X_k$ be independent random variables with $X_i\in [0,1]$ for $i\in [k]$ and let $X=\sum_{i\in [k]}{X_i}$. Then, $\E[X] \leq \frac{\E[\min\{1,X\}]}{1-\E[\min\{1,X\}]}$.
\end{lemma} 

\begin{proof} 
Using Lemma~\ref{lem:min} and applying the geometric-arithmetic mean inequality, we get
\begin{align*}
    \E[\min\{1,X\}] &\geq 1-\prod_{i\in [k]}{\left(1-\E[X_i]\right)}\geq 1-\left(\frac{1}{k}\cdot \sum_{i\in [k]}{(1-\E[X_i])}\right)^{k}=1-\left(1-\frac{1}{k}\cdot \E[X]\right)^{k},
\end{align*} 
which implies that
\begin{align*}
    \E[X] &\leq k\cdot \left(1-(1-\E[\min\{1,X\}])^{1/k}\right).
\end{align*}
By the concavity of function $z^{1/k}$, we get $z_1^{1/k}-z_2^{1/k}\leq \frac{1}{k}\cdot (z_1-z_2)\cdot z_1^{1/k-1}$ for $z_2<z_1$. Applying this inequality to the RHS above with $z_1=1$ and $z_2=1-\E[\min\{1,X\}]$, we get
\begin{align*}
    \E[X] &\leq \E[\min\{1,X\}]\cdot (1-\left(1-\E[\min\{1,X\}])^{1/k-1}\right)\\
    &\leq \E[\min\{1,X\}]\cdot (1-\left(1-\E[\min\{1,X\}])^{-1}\right)\\
    &=\frac{\E[\min\{1,X\}]}{1-\E[\min\{1,X\}]},
\end{align*}
as desired.
\end{proof}

Lemma~\ref{lem:min-new} clearly implies that $\E[X]\leq \frac{1}{p-1}$ when $\E[\min\{1,X\}]\leq 1/p$. Lemma~\ref{lem:relaxed-constraint} now follows by applying Lemma~\ref{lem:min-new} with $k=n$ and the random variable $X_i=D(v_i,p)$ for $i\in [n]$, where $v_i$ is drawn from the probability distribution $\F_i$. Notice that $\E_{v_i\in \F_i}[D(v_i,p)]=\int_0^1{(1-F_{i,x}(p))\ud x}$ and, hence, $\E[X]=\sum_{i\in [n]}{X_i}=\sum_{i\in [n]}{\int_0^1{(1-F_{i,x}(p))\ud x}}$, while $\E[\min\{1,X\}]\leq 1/p$ is due to constraint (\ref{eq:constraint-min}) of the mathematical program (\ref{eq:mp-original}).
\end{proof}

Our next lemma shows that the mathematical program (\ref{eq:mp-original}) has an approximately optimal feasible solution in which no selling probability is very close to $0$.

\begin{lemma}\label{lem:pumped-up-q-feasible}
The mathematical program (\ref{eq:mp-original}) has a feasible solution $(\widetilde{\q},\F)$ which furthermore satisfies $\widetilde{q}_i(x)\geq \frac{1}{2n}$ for every $i\in [n]$ and $x\in [0,1]$ and approximates its maximum objective value within a factor of $2$.
\end{lemma}

\begin{proof}
Consider the optimal feasible solution $(\q,\F)$ of the mathematical program (\ref{eq:mp-original}) and define $\widetilde{q}_i(x)=\frac{1}{2n}+\frac{q_i(x)}{2}$ for $i\in [n]$ and $x\in [0,1]$. The only constraint of the mathematical program (\ref{eq:mp-original}) that can be affected by the change from $\q$ to $\widetilde{\q}$ is (\ref{eq:constraint-bounded-q}). Due to the feasibility of solution $(\q,\F)$, we get
    \begin{align*}
        \sum_{i\in [n]}{\int_0^1{\widetilde{q}_i(x)\ud x}} &= \sum_{i\in [n]}{\int_0^1{\left(\frac{1}{2n}+\frac{q_i(x)}{2}\right)\ud x}} =\frac{1}{2}+\frac{1}{2}\cdot \sum_{i\in [n]}{\int_0^1{q_i(x)\ud x}}\leq 1,
    \end{align*}
    implying that solution $(\widetilde{\q},\F)$ is feasible as well.

 To prove that the objective value $(\widetilde{\q},\F)$ approximates $(\q,\F)$ within $2$, it suffices to prove that $2\cdot R_{\widetilde{q}_i,\F_i}(x)\geq R_{q_i,\F_i}(x)$ for every $i\in [n]$ and $x\in [0,1]$. Let $i\in [n]$ and $x\in [0,1]$. The claim is obvious if $q_i(x)=1/n$ (since $\widetilde{q}_i(x)=1/n$ as well). We distinguish between two further cases.

 \paragraph{Case 1: $q_i(x)<1/n$.} Notice that $\widetilde{q}_i(x)>q_i(x)$ in this case. Recall that due to the regularity of probability distribution $F_{i,x}$, the revenue-quantile curve $z\cdot F_{i,x}^{-1}(1-z)$ is concave. Thus, it takes values above the line $q(z)=\frac{q_i(x)}{1-q_i(x)}\cdot F_{i,x}^{-1}(1-q_i(x))\cdot (1-z)$ connecting points $(q_i(x),q_i(x)\cdot F_{i,x}^{-1}(1-q_i(x)))$ and $(1,0)$. Hence,
    \begin{align*}
        2\cdot R_{\widetilde{q}_i,\F_i}(x) &=2\cdot \left(\frac{1}{2n}+\frac{q_i(x)}{2}\right)\cdot F_{i,x}^{-1}\left(1-\frac{1}{2n}-\frac{q_i(x)}{2}\right)\\
        &\geq 2\cdot \frac{q_i(x)}{1-q_i(x)}\cdot F_{i,x}^{-1}(1-q_i(x))\cdot \left(1-\frac{1}{2n}-\frac{q_i(x)}{2}\right)\\
        &\geq q_i(x)\cdot F_{i,x}^{-1}(1-q_i(x))=R_{q_i,\F_i}(x),
    \end{align*}
as desired. 

\paragraph{Case 2: $q_i(x)>1/n$.} Notice that $\widetilde{q}_i(x)<q_i(x)$ in this case. Again, due to the regularity of probability distribution $F_{i,x}$, the revenue-quantile curve $z\cdot F_{i,x}^{-1}(1-z)$ is concave and takes values above the line $g(z)=F_{i,x}^{-1}(1-q_i(x)))\cdot z$ connecting points $(0,0)$ and $(q_i(x),q_i(x)\cdot F_{i,x}^{-1}(1-q_i(x)))$. Hence,
   \begin{align*}
        2\cdot R_{\widetilde{q}_i,\F_i}(x) &=2\cdot \left(\frac{1}{2n}+\frac{q_i(x)}{2}\right)\cdot F_{i,x}^{-1}\left(1-\frac{1}{2n}-\frac{q_i(x)}{2}\right)\\
        &\geq 2\cdot F_{i,x}^{-1}(1-q_i(x))\cdot \left(\frac{1}{2n}+\frac{q_i(x)}{2}\right)\\
        &\geq q_i(x)\cdot F_{i,x}^{-1}(1-q_i(x))=R_{q_i,\F_i}(x),
    \end{align*}
    completing the proof.
\end{proof}

The proof of the next lemma follows a similar structure to the proof of Lemma~\ref{lem:imotv-single-agent} in Section~\ref{sec:single-agent}. However, due to the constraint (\ref{eq:constraint-bounded-q}) on the selling probabilities, the simple argument about their optimality that we used in the proof of Lemma~\ref{lem:imotv-single-agent} cannot be used here; the dependency on the number of agents in our analysis is mainly introduced due to this subtle issue.

\begin{lemma}\label{lem:pumped-up-q-imotv-bound}
Let $(\widetilde{\q},\F)$ be any feasible solution of the  mathematical program (\ref{eq:mp-original}) that satisfies the additional constraint $\widetilde{q}_i(x)\geq \frac{1}{2n}$ for every $i\in [n]$ and $x\in [0,1]$. Then, for every $i\in [n]$, it holds 
    \begin{align}
        R_{\widetilde{q}_i,\F_i}(0)\leq 8n\kappa^2\int_0^1{R_{\widetilde{q}_i,\F_i}(x)\ud x}.
    \end{align}
\end{lemma}

\begin{proof}
Consider an agent $i\in [n]$ and define $S_i=\left\{v_i\in \calV_\kappa: v'_i(0)\geq F_{i,0}^{-1}(1-\widetilde{q}_i(0))\right\}$ to be the set of valuation functions drawn from $\F_i$ with derivative at least $F_{i,0}^{-1}(1-\widetilde{q}_i(0))$ at point $0$. By the definition of set $S_i$, we have $\Pr_{v\sim \F_i}[v\in S_i]=1-\widetilde{q}_i(0)$. Hence, 
    \begin{align}\label{eq:step1-multi-agent}
        \min_{v\in S_i}{v'(x)}& \leq F_x^{-1}(1-\widetilde{q}_i(0))\leq F_x^{-1}(1-\widetilde{q}_i(x)) \leq \frac{2n}{\widetilde{q}_i(0)} \cdot \widetilde{q}_i(x)\cdot F_x^{-1}(1-\widetilde{q}_i(x)).
    \end{align}
    The second inequality follows due to the constraint for monotone non-increasing selling probability in mathematical program (\ref{eq:mp-original}), and the third one since $\widetilde{q}_i(x)\geq \frac{1}{2n}$ and $\widetilde{q}_i(0)\leq 1$. By Equation (\ref{eq:step1-multi-agent}), we get
    \begin{align}\label{eq:relate-min-multi-agent}
        \int_0^1{\min_{v\in S_i}{v'(x)}\ud x} &\leq \frac{2n}{\widetilde{q}(0)}\cdot \int_0^1{R_{\widetilde{q}_i,\F_i}(x)\ud x}.
    \end{align}
    Now, let $t\in [0,1]$ and $\widehat{v}$ be a valuation function of $S_i$ satisfying $\widehat{v}'(t)=\min_{v\in S_i}{v'(t)}$. We have
    \begin{align}\nonumber
        \int_0^1{\widehat{v}'(x)\ud x} &= \int_0^t{\widehat{v}'(x)\ud x}+\int_t^1{\widehat{v}'(x)\ud x}\\\nonumber
        &\leq t\cdot \widehat{v}'(0)+(1-t)\cdot \widehat{v}'(t)\\\nonumber
        &\leq t\cdot \widehat{v}'(0)+\frac{1-t}{t}\cdot \int_0^t{\min_{v\in S}{v'(x)}\ud x}\\\label{eq:compare-imotv-step2-multi-agent}
        &\leq t\cdot \widehat{v}'(0)+ \frac{2n\cdot (1-t)}{t\cdot \widetilde{q}_i(0)}\cdot \int_0^1{R_{\widetilde{q}_i,\F_i}(x)\ud x}.
    \end{align}
    The first two inequalities are due to the concavity of the valuation function $\widehat{v}$ (and its definition), while the third one is due to Equation (\ref{eq:relate-min-multi-agent}). Now, we divide both sides of Equation (\ref{eq:compare-imotv-step2-multi-agent}) by $\widehat{v}'(0)$ and use the facts that the valuation function $\widehat{v}$ has IMOTV at most $\kappa$ and belongs to set $S_i$ (and, therefore, $\widehat{v}'(0)\geq F_{i,0}^{-1}(1-\widetilde{q}_i(0)$). We obtain
    \begin{align}\label{eq:compare-imotv-step3-multi-agent}
        \frac{1}{\kappa} &\leq \frac{\int_0^1{\widehat{v}'(x)\ud x}}{\widehat{v}'(0)} \leq t +\frac{1-t}{t}\cdot \frac{\int_0^1{R_{q,\F}(x)\ud x}}{\widetilde{q}(0)\cdot \widehat{v}'(0)} \leq t+\frac{2n\cdot (1-t)}{t}\cdot \frac{\int_0^1{R_{\widetilde{q}_i,\F_i}(x)\ud x}}{R_{\widetilde{q}_i,\F_i}(0)}.
    \end{align}
    Rearranging Equation (\ref{eq:compare-imotv-step3-multi-agent}), we get
    \begin{align}\label{eq:compare-imotv-step4-multi-agent}
        \frac{R_{\widetilde{q}_i,\F_i}(0)}{\int_0^1{R_{\widetilde{q}_i,\F_i}(x)\ud x}} &\leq \frac{2n\cdot(1-t)}{t\cdot \left(\frac{1}{\kappa}-t\right)}.
    \end{align}
    Similarly to the last step in the proof of Lemma~\ref{lem:imotv-single-agent}, setting $t=1-\sqrt{1-\frac{1}{\kappa}}$, Equation (\ref{eq:compare-imotv-step4-multi-agent}) yields
    \begin{align*}
        \frac{R_{\widetilde{q}_i,\F_i}(0)}{\int_0^1{R_{\widetilde{q}_i,\F_i}(x)\ud x}} &\leq 2n\cdot \left(1-\sqrt{1-\frac{1}{\kappa}}\right)^{-2}\leq 8n\kappa^2,
    \end{align*}
    as desired. 
\end{proof}

The next two lemmas refer to optimal feasible solutions of the mathematical program (\ref{eq:mp-original}), among those satisfying the additional constraint.
\begin{lemma}\label{lem:monotone-R}
    Let $(\widetilde{\q},\F)$ be an optimal feasible solution to mathematical program (\ref{eq:mp-original}), under the additional constraint $\widetilde{q}_i(x)\geq \frac{1}{2n}$ for $i\in [n]$ and $x\in [0,1]$. Then, $R_{\widetilde{q}_i,\F_i}(x)$ is monotone non-decreasing in $x$ for every $i\in [n]$.
\end{lemma}

\begin{proof}
For the sake of contradiction, let $0\leq x_1<x_2\leq 1$ and assume that $R_{\widetilde{q}_i,\F_i}(x)<R_{\widetilde{q}_i,\F_i}(x_2)$ for every $x\in [x_1,x_2)$. Then, notice that due to the concavity of the valuation functions, we have
\begin{align*}
    \widetilde{q}_i(x_2)\cdot F_{i,x}^{-1}(1-\widetilde{q}_i(x_2)) \geq \widetilde{q}_i(x_2)\cdot F_{i,x_2}^{-1}(1-\widetilde{q}_i(x_2)) > \widetilde{q}_i(x_1)\cdot F_{i,x_1}^{-1}(1-\widetilde{q}_i(x_1)).
\end{align*}
This contradicts the optimality of the solution $(\widetilde{\q},\F)$ since decreasing $\widetilde{q}_i(x)$ to $\widetilde{q}_i(x_2)$ for $x\in [x_1,x_2)$ would strictly increase the objective value of the mathematical program (\ref{eq:mp-original}), without violating constraint (\ref{eq:constraint-bounded-q}), the constraint for monotone non-increasing selling probability, as well as the additional constraint (recall that $\widetilde{q}_i(x_2)\geq \frac{1}{2n}$).
\end{proof}

\begin{lemma}\label{lem:bound-on-R}
Let $(\widetilde{\q},\F)$ be an optimal feasible solution to mathematical program (\ref{eq:mp-original}), under the additional constraint $\widetilde{q}_i(x)\geq \frac{1}{2n}$ for $i\in [n]$ and $x\in [0,1]$. Then, $R_{\widetilde{q}_i,\F_i}(x)\leq 256n^2\kappa^4$ for every $i\in [n]$ and $x\in [0,1]$.
\end{lemma}

\begin{proof}
   Let $i\in [n]$ and $z\in [0,1]$; for the linear price $p=F_{i,z}^{-1}(1-\widetilde{q}_i(z))$, the constraint (\ref{eq:constraint-min}) of the mathematical program (\ref{eq:mp-original}) yields
    \begin{align*}
        \frac{1}{F_{i,z}^{-1}(1-\widetilde{q}_i(z))} &\geq \E_{\vv\sim\F}\left[\min\left\{1,\sum_{i\in [n]}{D(v_i,F_{i,z}^{-1}(1-\widetilde{q}_i(z)))}\right\}\right]\\
        &\geq \E_{v_i\sim\F_i}\left[D(v_i,F_{i,z}^{-1}(1-\widetilde{q}_i(z)))\right]\geq z\cdot \widetilde{q}_i(z).
    \end{align*}
    The last inequality is due to the fact that, by definition, the demand $D(v_i,F_{i,z}^{-1}(1-\widetilde{q}_i(z)))$ is at least $z$ with probability $\widetilde{q}_i(z)$. Hence,
    \begin{align}\label{eq:bound-for-integration}
        R_{\widetilde{q}_i,\F_i}(z) &= \widetilde{q}_i(z)\cdot F_{i,z}^{-1}(1-\widetilde{q}_i(z))\leq \frac{1}{z},
    \end{align}
    for every $z\in [0,1]$.

    Now, let $t\in [0,1]$. Using Lemma~\ref{lem:pumped-up-q-imotv-bound}, Equation (\ref{eq:bound-for-integration}), and the inequality $R_{\q_i,\F_i}(0)\geq R_{\widetilde{q}_i,\F_i}(x)$ from Lemma~\ref{lem:monotone-R}, we have
    \begin{align*}
        \frac{R_{\widetilde{q}_i,\F_i}(0)}{8n\kappa^2} & \leq \int_0^1{R_{\widetilde{q}_i,\F_i}(x)\ud x} = \int_0^t{R_{\widetilde{q}_i,\F_i}(x)\ud x}+\int_t^1{R_{\widetilde{q}_i,\F_i}(x)\ud x}\\
        &\leq t\cdot R_{\widetilde{q}_i,\F_i}(0) + \int_t^1{\frac{\ud x}{x}}= t\cdot R_{\widetilde{q}_i,\F_i}(0)+\ln{\frac{1}{t}}.
    \end{align*}
    Setting $t=\left(R_{\widetilde{q}_i,\F_i}(0)\right)^{-1}$ and using the inequality $1+\ln{z}\leq 2\sqrt{z}$, we get
    \begin{align*}
        \frac{R_{\q_i,\F_i}(0)}{8n\kappa^2} & \leq 1+\ln{R_{\widetilde{q}_i,\F_i}(0)}\leq 2\cdot \sqrt{R_{\widetilde{q}_i,\F_i}(0)}
    \end{align*}
    and, equivalently, $R_{\widetilde{q}_i,\F_i}(0)\leq 256n^2\kappa^4$. The lemma now follows due to Lemma~\ref{lem:monotone-R}.
\end{proof}

To summarize the proof of Theorem~\ref{thm:rev-main-direct} so far, Lemmas~\ref{lem:relaxed-constraint},~\ref{lem:pumped-up-q-feasible}, and~\ref{lem:bound-on-R} imply that the mathematical program (\ref{eq:mp-original}) has a feasible solution $(\q,\F)$ that approximates its maximum objective value within a factor of $2$ and satisfies the next two sets of inequalities:
\begin{align}\label{eq:relaxed-linear-pricing-constraint}
&\sum_{i\in [n]}{\int_0^1{(1-F_{i,x}(p))\ud x}}\leq \frac{1}{p-1}, \quad\forall p>1\\\label{eq:bounded-R-restated}
&R_{q_i,\F_i}(x)\leq 256n^2\kappa^4, \quad\forall i\in [n], x\in 0,1]\\\nonumber
\end{align}
Using this solution, we define the  subsets $A$, $B$, and $C_t$ for $t\geq 2$ of set $[n]\times [0,1]$:
\begin{align*}
    & A=\left\{(i,x)\in [n]\times [0,1] \mbox{ s.t.~} F_{i,x}^{-1}(1-q_i(x))<2\right\}\\
    & B=\left\{(i,x)\in [n]\times [0,1] \mbox{ s.t.~} F_{i,x}^{-1}(1-q_i(x))\geq 2 \mbox{ and } q_i(x)\cdot F_{i,x}^{-1}(1-q_i(x))< 2\cdot (1-q_i(x)) \right\}\\
    & C_t=\left\{(i,x)\in [n]\times [0,1] \mbox{ s.t.~} F_{i,x}^{-1}(1-q_i(x))\geq t \mbox{ and } q_i(x)\cdot F_{i,x}^{-1}(1-q_i(x))\geq t \cdot (1-q_i(x)) \right\}
\end{align*}
Furthermore, we define
\begin{align*}
    \I\left(g_i(x);(i,x)\in X\right) = 2\cdot \sum_{i\in [n]}{\int_0^1{g_i(x)\one\{(i,x)\in X\}\ud x}}
\end{align*}
for functions $g_i: [0,1]\rightarrow \R_{\geq 0}$ for $i\in [n]$ and $X\subseteq [n]\times [0,1]$.
Thus, the maximum objective value of the mathematical program (\ref{eq:mp-original}) is upper-bounded by
\begin{align}\nonumber
    &2\cdot \sum_{i\in [n]}{\int_0^1{q_i(x)\cdot F_{i,x}^{-1}(1-q_i(x))\ud x}} = \I(R_{q_i,\F_i}(x); (i,x)\in [m]\times [0,1])\\\nonumber
    &\quad\quad= \I(R_{q_i,\F_i}(x); (i,x)\in A)+\I(R_{q_i,\F_i}(x); (i,x)\in B)\\\label{eq:summands-I}
    &\quad\quad\quad\quad +\sum_{j=1}^{\left\lfloor\log{( 256n^2\kappa^4)}\right\rfloor}{\I(R_{q_i,\F_i}(x); (i,x)\in C_{2^j}\setminus C_{2^{j+1}})}.
\end{align}
The last sum runs up to $\lfloor\log{\left(256n^2\kappa^4\right)}\rfloor$ since, due to the inequalities (\ref{eq:bounded-R-restated}), the set $C_t$ is empty for $t>256n^2\kappa^4$. We will complete the proof of Theorem~\ref{thm:rev-main-direct} by upper-bounding each summand in the RHS of Equation (\ref{eq:summands-I}) by a constant. In doing so, we will be assisted by a last technical lemma.

\begin{lemma}\label{lem:triangular}
Let $F$ be a regular probability distribution, $z\in [0,1]$, and $\ell>0$ such that $\ell\leq F^{-1}(1-z)$. Then,
\begin{align*}
    \frac{z\cdot F^{-1}(1-z)}{(1-z)\cdot \ell+z\cdot  F^{-1}(1-z)} \leq 1-F(\ell).
\end{align*}
\end{lemma}

\begin{proof}
We use lower-left and upper-right points to define axis-parallel rectangles. Notice that the assumption $\ell\leq F^{-1}(1-z)$ implies that $z\leq 1-F(\ell)$. We distinguish between two cases. 

\paragraph{Case 1: $\ell\cdot (1-F(\ell)) \geq z\cdot F^{-1}(1-z)$.} Then, the area of the rectangle defined by the points $(0,0)$ and $(1,z\cdot F^{-1}(1-z))$ is upper-bounded by the total area in the two rectangles defined by the points $(0,0)$ and $(1-F(\ell),z\cdot F^{-1}(1-z))$ and by the points $(z,0)$ and $(1,\ell\cdot (1-F(\ell)))$, respectively. Thus,
\begin{align*}
    z\cdot F^{-1}(1-z) &\leq (1-F(\ell))\cdot z\cdot F^{-1}(1-z)+(1-z)\cdot \ell\cdot (1-F(\ell)),
\end{align*}
which is equivalent to the claimed inequality.

\begin{figure}[h]

    \centering
    \includegraphics[width=0.6\textwidth]{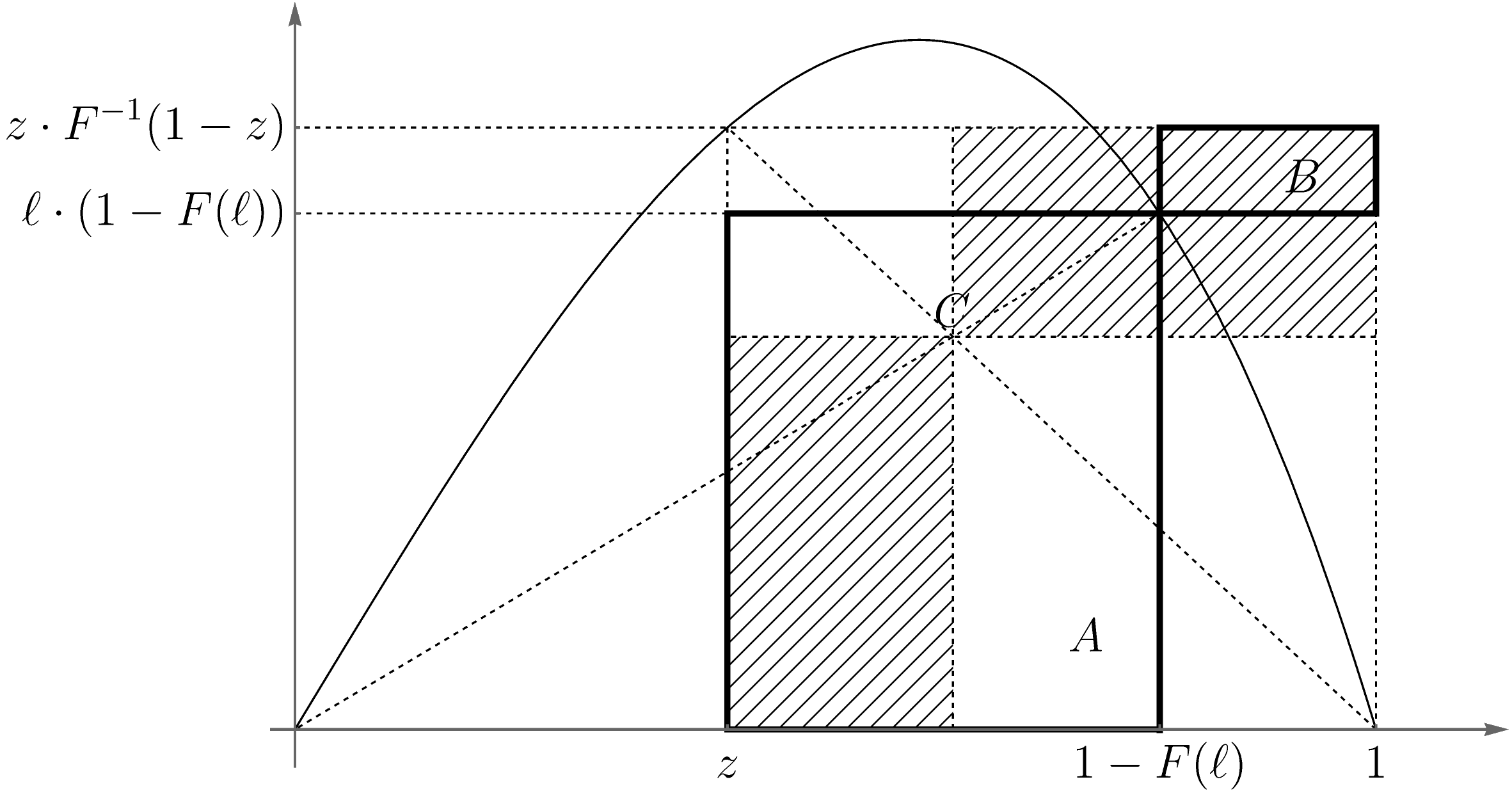}
    
    \caption{An illustration of the argument in the proof of Lemma~\ref{lem:triangular} (Case 2) using the revenue-quantile curve $q\cdot F^{-1}(1-q)$. To compare the two rectangles with thick sides $A$ and $B$, we show that due to concavity of the revenue curve $q\cdot F^{-1}(1-q)$, rectangle $A$ contains the lower-left shaded region corresponding to rectangle $\widetilde{A}$ and the upper-right shaded region corresponding to rectangle $\widetilde{B}$ contains rectangle $B$. To complete the proof, it suffices to observe that the two shaded rectangles have equal areas. }
    \label{fig:triangle}
\end{figure}

\paragraph{Case 2: $\ell\cdot (1-F(\ell)) < z\cdot F^{-1}(1-z)$.} Denote by $A$ the axis-parallel rectangle defined by points $(z,0)$ and $(1-F(\ell),\ell\cdot (1-F(\ell)))$ and observe that its area is equal to $\ell\cdot (1-F(\ell))\cdot (1-F(\ell)-z)$. Also, denote by $B$ the axis-parallel rectangle defined by points $(1-F(\ell),\ell\cdot (1-F(\ell)))$ and $(1,z\cdot (F^{-1}(1-z))$ and observe that its area is equal to $F(\ell)(z\cdot F^{-1}(1-z)-\ell\cdot(1-F(\ell))$. We claim that rectangle $A$ is at least as large as rectangle $B$ and, hence,
\begin{align*}
    &z\cdot F^{-1}(1-z) \leq z\cdot F^{-1}(1-z)+\area(A)-\area(B)\\
    &=z\cdot F^{-1}(1-z)+(1-F(\ell)-z)\cdot \ell\cdot (1-F(\ell))-F(\ell)(z\cdot F^{-1}(1-z)-\ell\cdot(1-F(\ell)))\\
    &=(1-F(\ell)\cdot (z\cdot F^{-1}(1-z)+\ell\cdot (1-z)), 
\end{align*}
which is equivalent to the stated inequality. The argument is illustrated in Figure~\ref{fig:triangle}. 

To prove the claim $\area(A)\geq \area(B)$, it suffices to consider the point $C$ at the intersection of the line connecting points $(z,z\cdot F^{-1}(1-z))$ and $(1,0)$ with the line connecting points $(0,0)$ and $(1-F(\ell),\ell\cdot (1-F(\ell)))$. Now, denote by $\widetilde{A}$ and $\widetilde{B}$ the axis-parallel rectangles defined by points $(z,0)$ and $C$ and by points $C$ and $(1,z\cdot F^{-1}(1-z))$, respectively. Due to the concavity of the revenue-quantile curve $q\cdot F^{-1}(1-q)$, the point $(1-F(\ell),\ell\cdot(1-F(\ell)))$ lies inside rectangle $\widetilde{B}$ and point $C$ lies inside rectangle $A$. Thus, $\area(A)\geq \area(\widetilde{A})$ and $\area(B)\leq \area(\widetilde{B})$.
Furthermore, observe that rectangles $\widetilde{A}$ and $\widetilde{B}$ have the same area (this would still hold even if we had picked point $C$ anywhere in the line connecting points $(z,z\cdot F^{-1}(1-z))$ and $(1,0)$). Thus, we have
\begin{align*}
    \area(A)-\area(B)\geq \area(\widetilde{A})-\area(\widetilde{B}) = 0
\end{align*}
as claimed.
\end{proof}

Now, we have
\begin{align*}
    \I(R_{q_i,\F_i}(x); (i,x)\in A) &\leq 2\cdot \I(q_i(x); (i,x)\in A) \leq 4.
\end{align*}
The first inequality follows by the definition of set $A$, and the second one by inequality (\ref{eq:relaxed-linear-pricing-constraint}). Also,
\begin{align*}
    \I(R_{q_i,\F_i}(x); (i,x)\in B) &\leq 4\cdot \I\left(\frac{q_i(x)\cdot F_{i,x}^{-1}(1-q_i(x))}{4\cdot (1-q_i(x))}; (i,x)\in B\right)\\
    &\leq 4\cdot \I\left(\frac{q_i(x)\cdot F_{i,x}^{-1}(1-q_i(x))}{2\cdot (1-q_i(x))+q_i(x)\cdot F_{i,x}^{-1}(1-q_i(x))}; (i,x)\in B\right)\\
    &\leq 4\cdot \I\left(1-F_{i,x}(2); (i,x)\in B\right)\leq 8.
\end{align*}
The first inequality is obvious, the second one follows by the definition of set $B$, the third one follows by applying Lemma~\ref{lem:triangular} with $F=F_{i,x}$, $z=q_i(x)$, and $\ell=2$, and the fourth one follows by applying inequality (\ref{eq:relaxed-linear-pricing-constraint}) with $p=2$. 
Finally, 
\begin{align*}
    &\I(R_{q_i,\F_i}(x); (i,x)\in C_{2^j} \setminus C_{2^{j+1}}) < \I\left(2^{j+1};(i,x)\in C_{2^j} \setminus C_{2^{j+1}}\right)\\
    &=2^{j+2}\cdot \I\left(\frac{1}{2};(i,x)\in C_{2^j} \setminus C_{2^{j+1}}\right)\\
    &\leq 2^{j+2}\cdot \I\left(\frac{q_i(x)\cdot F_{i,x}^{-1}(1-q_i(x))}{2^j\cdot(1-q_i(x))+q_i(x)\cdot F_{i,x}^{-1}(1-q_i(x))};(i,x)\in C_{2^j} \setminus C_{2^{j+1}}\right)\\
    &\leq 2^{j+2}\cdot \I\left(1-F_{i,x}(2^j);(i,x)\in C_{2^j} \setminus C_{2^{j+1}}\right)\\
    &\leq \frac{2^{j+3}}{2^j-1}\leq 16.
\end{align*}
The first inequality follows since $(i,x)\not\in C_{2^{j+1}}$ and, hence, $q_i(x)\cdot F_{i,x}^{-1}(1-q_i(x))<\max\{2^{j+1} \cdot(1-q_i(x),2^{j+1}\cdot q_i(x)\}\leq 2^{j+1}$. The second inequality follows since $(i,x)\in C_{2^j}$ and, hence, $q_i(x)\cdot F_{i,x}^{-1}(1-q_i(x))\geq 2^j\cdot (1-q_i(x))$. The third inequality follows by applying Lemma~\ref{lem:triangular} with $F=F_{i,x}$, $z=q_i(x)$, and $\ell=2^j$, and the fourth one by inequality (\ref{eq:relaxed-linear-pricing-constraint}) for $p=2^j$. 

Overall, the RHS of Equation (\ref{eq:summands-I}) and, consequently, the maximum objective value of mathematical program (\ref{eq:mp-original}) and the revenue gap, is upper-bounded by $4+8+16\cdot \lfloor \log{\left(256n^2\kappa^4\right)}\rfloor\leq 140+32\cdot\log{n}+64\cdot \log{\kappa}\in \bigO(\ln\kappa+\ln{n})$. The proof of Theorem~\ref{thm:rev-main-direct} is now complete.
\end{proof}

\section{Extensions and open problems}\label{sec:open}
Motivated by recent work on ``simple vs.~optimal'' revenue trade-offs in Bayesian mechanism design, we have explored the power of linear posted pricing when selling a perfectly divisible item. Our work leaves several interesting open problems. Of course, we would like to determine the tight bound on the revenue gap for linear pricing, closing the gap between our upper bound of $\bigO\left(\min\left\{\ln{\kappa}+\ln{n},\kappa^2\right\}\right)$ and our lower bound of $\Omega(\ln{\kappa})$. It is tempting to conjecture that the dependency on the number of agents is not necessary and $\Theta(\ln{\kappa})$ is the tight bound on the revenue gap. Furthermore, we would like to explore whether the revenue gap increases significantly when agents are utility-maximizing instead of shortsighted. Preliminary attempts to show a separation between the two agent models have not been successful. Finally, the optimal revenue for the ex-ante relaxation of the multi-parameter setting in Section~\ref{subsec:rev:mp} could be used to bound the revenue gap of other mechanisms for allocating a divisible item, such as the proportional mechanism~\cite{JT04,K97}, which has been extensively studied in terms of social welfare~\cite{CV16,CV18,CST16}. This could lead to results on the Bayes-Nash price of anarchy for revenue which, with a few exceptions such as \cite{CKKK14,HHT14,LLT12}, are rather sporadic in the literature.

\bibliographystyle{plain}
\bibliography{pricing.arxiv}

\end{document}